\documentclass[11pt,a4paper]{article}
\usepackage{latexsym,amssymb,amsmath,amsthm,amsfonts,enumerate,verbatim,xspace,exscale}

\usepackage{graphicx}
\usepackage{amsmath,amsbsy,amsfonts,amssymb}
\usepackage{subfigure}

\theoremstyle{plain}
\newtheorem{theorem}{Theorem}[section]

\newtheorem{proposition}[theorem]{Proposition}

\newtheorem*{theorem*}{Theorem}

\theoremstyle{definition}

\def\vv<#1>{\langle#1\rangle}

\def\R{\mathbb{R}}
\def\I{\mathbb{I}}
\def\Lag{\mathcal{L}}

\newcommand{\g}{\mathfrak{g}}
\def\se{se}

 \newcommand{\ad}{\mbox{$\text{\upshape{ad}}$}}

\newcommand{\vecu}{{\bf u}}
\newcommand{\In}{\mathcal{I}}
\newcommand{\Jn}{\mathcal{K}}

\def\F{\mathcal{F}}
\def\M{\mathcal{M}}
\def\B{\mathcal{B}}
\def\vecom{\boldsymbol{\omega}}
\def\sech{\operatorname{sech}}

\title{The Hydrodynamic Chaplygin Sleigh}
\author{Yuri N. Fedorov \\
 Department de Matem\'atica Aplicada I, \\
Universitat Politecnica de Catalunya, \\
Barcelona, E-08028 Spain \\
e-mail: Yuri.Fedorov@upc.edu \\
and \\
Luis C. Garc\'ia-Naranjo  \\
Section de Mathematiques \\
Station 8, EPFL \\
CH-1015 Lausanne, Switzerland \\
e-mail: luis.garcianaranjo@epfl.ch}

\begin{document}
\maketitle

\begin{abstract}
We consider the motion of rigid bodies in a potential fluid subject to certain nonholonomic constraints and show
that it is described by  Euler--Poincar\'e--Suslov equations.
In the 2-dimensional case, when the constraint is realized by a blade attached to the body,
the system provides a hydrodynamic generalization of the Chaplygin sleigh,
whose dynamics are studied in detail.
Namely, the equations of motion are integrated explicitly and the asymptotic behavior of the system is determined.
It is shown how the presence of the fluid brings new features to such a behavior.
\end{abstract}


\section{Introduction and outline}

This paper considers the motion of rigid bodies in a potential fluid
in the presence of certain kind of linear nonholonomic constraints. This is motivated by studying the dynamics of underwater vehicles
with large fins, which, in the first approximation,
impose restrictions on the velocity of their central points relative to the fluid.

The free motion of such vehicles in a 3-dimensional potential fluid
can be described by the finite-dimensional Kirchhoff equations, in which the action of the
fluid is reflected via the tensor of adjoint masses that depends solely on the body shape. In the presence of a large fin this tensor
becomes almost singular: one of its eigenvalues is very large. 
As was shown in
\cite{Koz2}, see also \cite{ArnoldIII}, section 6.4, in the limit the dynamics of the body with the fin
is described by a {\it vakonomic} system
with the nonholonomic constraint that prohibits the instantaneous motion in the direction of the first eigenvalue.

It should be noted however, that the vakonomic description of several mechanical systems with nonholonomic
constraints   gives rise to rather unexpected behavior (see again \cite{Koz2}).

In our paper we follow a more direct approach. Namely, we describe the motion of the body with a fin by the Kirchhoff equations with a regular
tensor of adjoint masses and impose linear constraint(s) on the velocities, which serve to model the action of the fin. The reaction forces arising from the constraint
are included according to the Lagrange-D'Alembert principle, which is widely accepted
to be physically meaningful.

Our model for the constraint and the reaction forces can also be derived following the
{\it anisotropic friction} approach for realizing constraints described in  \cite{NeFu,Koz2}, see also \cite{Carath}.
In this approach one considers the unconstrained system under the influence of
a viscous frictional term  that only acts in the direction
perpendicular to the fin and is proportional to a parameter, say $\chi$. The equations of motion are then obtained by letting $\chi \rightarrow \infty$.
The obtained system is in general different from the limit vakonomic system described above.

The idea of modeling  the action of  the fin on the fluid with a nonholonomic constraint
using Lagrange-D'Alembert's principle also appears  in \cite{Rand}. Although it is regarded as an idealized situation that is physically unattainable, it is the first approximation
for their analysis.  The authors use it to study the motion of an underwater projectile with tail fins moving at
high speed.  However, due to the  high speed of the projectile,  cavitation effects appear and  the general interaction of the fluid and the body motion is not  modeled
with Kirchhoff's approach.

We thus believe that the systems that we present  serve as a first approximation  for the motion of an
underwater vehicle with a large (or very effective) fin that moves in a fluid
in the regime where  Kirchhoff's approach is valid.

In the case of a 2-dimensional body on a plane and 2-dimensional fluid, another motivation for the constraint appears: the body can interact with the plane via a sharp blade. This
setting gives a hydrodynamic generalization of the famous nonholonomic Chaplygin sleigh problem considered in detail in \cite{Chaplygin, NeFu}.

\paragraph{Contents of the paper.}
In Section 2 we review the preliminaries that are necessary for writing down
the equations of motion for our family of systems. We briefly recall  Kirchhoff's equations for a rigid body moving in a potential flow and
the Euler-Poincar\'e-Suslov equations for nonholonomic systems on Lie groups with left-invariant Lagrangian and  constraints.
Towards the end of the section the equations of motion for underwater
bodies subject to a Suslov or Chaplygin sleigh type constraint are given. We  then discuss the necessary and sufficient conditions for the existence of an invariant measure in a simple case.

The problem of the motion of the Chaplygin sleigh moving  in a potential flow (with no circulation) is treated in detail in Section 3. The reduced equations are written down explicitly for a general body
shape and their qualitative behavior is determined.  It is shown that  these equations
are Hamiltonian  with respect to a given bracket but do not preserve
a measure with a smooth density in the generic case.

We then continue to show that in the presence of
the fluid, the  sleigh generically evolves from one asymptotic circular motion to another in the opposite direction, although the limit circles do not coincide.
Their radii are given in terms of the components of the total inertia of the
fluid-body system.  For the purpose of concreteness, the
added inertia tensor is computed explicitly for an elliptical sleigh whose contact point $P$ with
the plane is located at the center of the ellipse $O$, and where the knife edge is not aligned with its
principal axes. This allows also to calculate the components of this tensor when $P$ does not coincide with $O$.

Then, in the general case, the asymptotic behavior (radius and course direction along the limit circle) is fully determined by the position of the center of mass.

In Section 4 the reduced equations of motion for the hydrodynamic Chaplygin sleigh
are integrated explicitly for a generic sleigh. The angular velocity
is integrated to give a closed expression for the
angle that determines the orientation of the body. Even though the position of the sleigh
cannot be obtained in a closed form, the distance between the centers of the limit circles is
computed explicitly.

Finally, in the Conclusions we motivate a further study of the hydrodynamic Chaplygin sleigh  in presence of  circulation and/or
point vortices.

\section{Preliminaries}
\label{S:Preliminaries}

\subsection*{Rigid body motion in a potential flow}

The motion of a rigid body in a
potential fluid in the absence of external forces was first described by Kirchhoff in 1890, who derived the reduced equations for the evolution of the body that do not incorporate the fluid itself. Its presence is instead encoded in
an ``added inertia" matrix that depends on the body shape.

Kirchhoff's equations can be understood as  the output of a two stage reduction procedure.
In the first stage one gets rid of the fluid variables by virtue of the ``particle relabeling symmetry".
In the second stage one eliminates the body configuration variables by  homogeneity
and isotropy of space. As such, Kirchhoff's equations are the Lie--Poisson
equations on the co-algebra $\se (3)^*$ where the  Hamiltonian is the total energy of
the fluid-body system. See \cite{Nov, Kanso} for more details on this reduction.

We briefly recall  Kirchhoff's equations in order to introduce the notation needed for the
rest of the paper. For a derivation of these equations obtained by
balancing the momentum of the body with the forces and  torques exerted by the fluid see Lamb's classic book on hydrodynamics,  \cite{Lamb}, which presents a thorough discussion of the problem and remains to date a key reference in the subject.

We adopt Euler's approach to the study of the rigid body dynamics and consider an orthonormal body frame
that is attached to the body. This frame is related to a fixed space frame
via an attitude (or rotation) matrix $g(t)\in SO(3)$. 

Let ${\bf V}(t)\in \R^3$, $\vecom (t) \in \R^3$ be the linear velocity of the origin of the body and its angular velocity, both vectors
are written with respect to the body frame. We then have
\begin{equation}
\label{E:Reconstruction_Equations}
\hat \vecom (t)=g^{-1}(t)\dot g(t), \qquad  {\bf V} (t)=g^{-1}(t)\dot {\bf x}(t),
\end{equation}
where the components of ${\bf x}(t)\in \R^3$ are the spatial coordinates of the origin of the body frame
at time $t$, and
\begin{equation*}
\hat \vecom := \left ( \begin{array}{ccc} 0 & - \omega_3 & \omega_2 \\ \omega_3 & 0 & -\omega_1 \\ - \omega_2 & \omega_1 & 0 \end{array} \right ),
\qquad \vecom =  \left ( \begin{array}{c}  \omega_1 \\ \omega_2  \\  \omega_3 \end{array} \right ), \quad
{\bf V} =  \left ( \begin{array}{c} v_1 \\ v_2  \\  v_3 \end{array} \right ) .
\end{equation*}

The configuration of the body is  completely determined by the pair $(g(t), {\bf x}(t))$, 
 an element of the Euclidean group $SE(3)$. 
In this way $\xi:=(\vecom, {\bf V})$ is  thought of as an element
of the Lie algebra $\se(3)$ that is identified with $\R^6$ via the bracket 
\begin{equation*}
[(\vecom , {\bf V}),(\boldsymbol{\eta}, {\bf U})]=(\vecom\times \boldsymbol{\eta}, \vecom \times {\bf U} - \boldsymbol{\eta}\times {\bf V}),
\end{equation*}
where ``$\times$" denotes the standard vector product in $\R^3$.

The kinetic energy of the body is given by (see, e.g., \cite{Landau})
\begin{equation}
\label{E:Body_Energy_in_terms_of_velocities}
T_\mathcal{B}= \frac{1}{2} \sum_{i=1}^3 \left ( m v_i^2 + \sum_{j=1}^3\In_\mathcal{B}^{ij}\omega_i\omega_j +2\Jn_\mathcal{B}^{ij}v_i\omega_j \right ),
\end{equation}
where $m$ is the total mass of the body and the constants $\In_\mathcal{B}^{ij}$ and $\Jn_\mathcal{B}^{ij}$, $i,j=1,2,3$, depend on its shape and the mass distribution. Here the $3\times 3$ matrix $\In_\mathcal{B}^{ij}$ is the usual inertia tensor of the body
with respect to the chosen frame. If the origin of the body frame is at the center of mass,  then $\Jn_\mathcal{B}^{ij}=0$.
For convenience we introduce the $6\times 6$ symmetric matrix
\begin{equation*}
\I_\mathcal{B}:= \left ( \begin{array}{cc} \In_{\mathcal{B}} & \Jn_\mathcal{B}  \\ \Jn_\mathcal{B}^T & mI \end{array} \right )
\end{equation*}
($I$ denoting the $3\times 3$ identity matrix) that defines $T_\mathcal{B}$ as a quadratic form on $\se(3)$. 

Next, the total energy of the fluid is given by
\begin{equation*}
T_\mathcal{F}=\frac{\rho}{2} \int
 || \vecu ||^2 \, dv,
\end{equation*}
where $\rho$ is the (constant) fluid density, $\vecu$ is the Eulerian velocity of the fluid,
and the integration takes place over the region occupied by the fluid.

We assume that the
fluid motion takes place in the boundless region of $\R^3$ that is not occupied by the body. We also assume that the flow is potential, with zero circulation, and is solely due to the motion of the body. Under these hypothesis
it is possible  express $T_\mathcal{F}$ as the quadratic form (see \cite{Lamb}):
\begin{equation}
\label{E:Fluid_Energy_in_terms_of_velocities}
T_\mathcal{F}= \frac{1}{2} \sum_{i,j=1}^3( \mathcal{M}_\mathcal{F}^{ij}v_iv_j +\In_\mathcal{F}^{ij}\omega_i\omega_j +2\Jn_\mathcal{F}^{ij}v_i\omega_j),
\end{equation}
where  $\mathcal{M}_\mathcal{F}^{ij}, \In_\mathcal{F}^{ij}$ and $\Jn_\mathcal{F}^{ij}$, $i,j=1,2,3$, are certain constants that only depend on the body shape.
They are referred to as \emph{added masses} and are conveniently written
in $6\times 6$ matrix form to define the \emph{added inertia tensor}:
\begin{equation*}
\I_\mathcal{F}:=\left ( \begin{array}{cc} \In_\mathcal{F} & \Jn_\mathcal{F} \\ \Jn_\mathcal{F}^T  & \mathcal{M}_\mathcal{F} \end{array} \right ),
\end{equation*}
where $\In_\mathcal{F}, \, \Jn_\mathcal{F}$, and  $\M_\mathcal{F}$ are the corresponding $3\times 3$ matrices.
One can show that the matrix $\I_\mathcal{F}$ is symmetric.

In the absence of potential forces, the total energy of the fluid-body system is $T=T_\B+T_\F$ and defines the kinetic energy Lagrangian
 $\Lag:T(SE(3))\rightarrow \R$. The motion of the body in space is determined by the geodesic motion
with respect to the Riemannian metric defined by $\Lag$.

 In view of
 (\ref{E:Body_Energy_in_terms_of_velocities}) and (\ref{E:Fluid_Energy_in_terms_of_velocities}), we can write the Lagrangian $\Lag=T_\B+T_\F$
in terms of the linear and angular velocities of the body (written in the body frame)  and this
expression does not  depend on the particular position and orientation of the body, i.e. is independent of $(g, {\bf x})$.
Thus $\Lag$ is  invariant under the lifted action  of left multiplication on $SE(3)$. This symmetry corresponds to invariance
under translations and rotations of the space frame. The reduction of this symmetry defines Euler-Poincar\'e equations on the
Lie algebra $\se(3)$ or, in the Hamiltonian setting, the Lie--Poisson equations
on the coalgebra $\se(3)^*$. The latter are precisely Kirchhoff's equations that are explicitly written below.

Define the reduced Lagrangian $L:\se(3)\rightarrow \R$ by
\begin{equation*}
 L(\xi)= \frac{1}{2} \xi^T \I \xi,
\end{equation*}
where $\xi=(\vecom, {\bf V})\in \R^3\times \R^3$ is thought as a  column vector and
the matrix $\I=\I_\mathcal{B}+\I_\mathcal{F}$.
 An element $\mu$ in the  co-algebra $\se(3)^*$ will be represented as a pair $\mu=({\bf k}, {\bf p})\in\R^3\times \R^3$ and will also be thought
of as a $6$ dimensional column vector.
Its action on $\xi= (\vecom, {\bf V})$ is defined by
\begin{equation}
\label{E:pairing_in_se(3)}
\langle \mu, \xi \rangle = {\bf k}\cdot \vecom +{\bf p} \cdot {\bf V},
\end{equation}
where $``\cdot "$ is the standard Euclidean scalar product. With this identification, the Legendre transform defines the mapping
between $\se(3)$ and $\se(3)^*$ given by $\mu=\I\xi$. Explicitly we have $\mu=({\bf k},{\bf p})$ where
\begin{equation}
\label{E:Explicit_Legendre}
{\bf k}=(\mathcal{I}_\mathcal{B}+\mathcal{I}_\mathcal{F})\vecom + (\Jn_\mathcal{B} + \Jn_\mathcal{F}) {\bf V}, \qquad
{\bf p}=m{\bf V}  + \mathcal{M}_\mathcal{F}{\bf V} +  ( \Jn_\mathcal{B}^T + \Jn_\mathcal{F}^T)\vecom.
\end{equation}
In classical hydrodynamics ${\bf k}$ and $\bf p$ are known as ``impulsive pair" and ``impulsive force" respectively.

The reduced Hamiltonian $H:\se(3)^*\rightarrow \R$ is given by $H(\mu)=\frac{1}{2}\mu^T\I^{-1}\mu$,
and the corresponding Lie--Poisson equations
$\dot \mu=\ad^*_{\I^{-1} \mu}\mu$ are then $(\dot {\bf k}, \dot {\bf p})=\ad^*_{(\vecom, {\bf V})}( {\bf k},  {\bf p})$. This gives the Kirchhoff equations
\begin{equation}
\label{E:Kirchhoff_1_3D}
\dot {\bf k}= {\bf k}\times \vecom + {\bf p} \times {\bf V}, \quad
\dot {\bf p} ={\bf p} \times \vecom .
\end{equation}

In the absence of the fluid ($\rho=0$) we have $\I_\mathcal{F}=0$.
Then, choosing the origin of the body axes at the center of mass, we obtain $\Jn_\B=0$ and ${\bf k}=\mathcal{I}_\mathcal{B}\vecom$, ${\bf p}=m{\bf V}$.
As a consequence, the equations (\ref{E:Kirchhoff_1_3D}) decouple, and we recover the well known fact about the motion of the body in vacuum: the center of mass
moves at constant velocity, whereas the body rotates freely according to the Euler equations.
It is also well known that in the presence of the fluid this is no longer true, that is, the fluid couples the translational
and rotational modes of the motion.

Given a solution of (\ref{E:Kirchhoff_1_3D}), the motion of the body in space is describing by solving
the reconstruction equations (\ref{E:Reconstruction_Equations}).

In the absence of circulation, the description of the motion of the body in the two-dimensional case is obtained in an analogous way.
The configuration space for the body motion is $SE(2)$ and we ultimately get Lie-Poisson equations on $\se(2)^*$.
This time we write $\xi\in \se(2)$ as $\xi=(\omega, {\bf V})\in \R\times \R^2$, and $\mu\in \se(2)^*$ as
$\mu=(k,{\bf p})\in \R\times \R^2$. The pairing between $\mu$ and $\xi$ is the analog of (\ref{E:pairing_in_se(3)}). Then
all of the above discussion for the three dimensional case remains
true by simply inserting the appropriate definition of the matrices  $\I_\mathcal{B}$ and $\I_\mathcal{F}$ that relate the column vectors $(k,{\bf p})$ and
$(\omega, {\bf V})$. These are $3\times 3$ matrices given by
\begin{equation}
\label{E:Inertia_Matrices_2D}
\I_\mathcal{B}=\left ( \begin{array}{ccc} \In + m(a^2+b^2) & -mb & ma \\ -mb   & m & 0 \\ ma& 0 & m \end{array} \right ), \qquad \I_\mathcal{F}=\left ( \begin{array}{cc} \In_\mathcal{F} & \Jn_\mathcal{F} \\ \Jn_\mathcal{F}^T  & \mathcal{M}_\mathcal{F} \end{array} \right ),
\end{equation}
where  $m$ is the mass of the body, $(a,b)$ are body coordinates of the center of mass, and $\In$ is the moment of inertia of the body about the center of mass. This time $  \In_\mathcal{F}$ is a scalar, $\Jn_\mathcal{F}$ is a two dimensional row vector, and $\mathcal{M}_\mathcal{F}$ is a $2\times 2$ matrix. As before,
the elements of $\I_\F$ depend solely
on the body shape.

 The Lie-Poisson equations $\dot \mu=\ad^*_{\I^{-1} \mu}\mu$ are given in components as 
\begin{equation*}
\begin{split}
\dot k &= v_2p_1-v_1p_2, \\
\dot p_1 &= \omega p_2, \quad  \quad
\dot p_2 = - \omega p_1,
\end{split}
\end{equation*}
where ${\bf p}=(p_1,p_2)$ and $(k,{\bf p})=\I (\omega, {\bf V})$ with $\I=\I_\B+\I_\F$.
The reconstruction equations (\ref{E:Reconstruction_Equations}) take the form
\begin{equation}
\label{E:Reconstruction_Equations_2D}
\dot \phi = \omega, \qquad v_1=\dot x \cos \phi  +  \dot y\sin \phi , \quad v_2= -\dot x\sin \phi   +  \dot y \cos \phi  ,
\end{equation}
where $\phi$ is the  rotation angle between the space and the body frame and $(x,y)$ are spatial coordinates of the origin of the body axes.

\subsection*{The Euler--Poincar\'e--Suslov equations.}

We have seen that Kirchhoff equations for a rigid body in a potential fluid are Lie-Poisson equations
on $\se(3)^*$ ($\se(2)^*$ in the two dimensional case) corresponding to a pure kinetic energy \emph{left} invariant
Lagrangian. We are interested in adding \emph{left} invariant nonholonomic constraints to the system. The resulting
reduced equations, that are consistent with Lagrange-D'Alembert's principle that states that the constraint
force performs  no work during the motion, are the so-called Euler--Poincar\'e--Suslov (EPS) equations.
We will write these equations explicitly.

In general, a nonholonomic system on a Lie group $G$ with a left invariant kinetic energy Lagrangian and
left invariant constraints is termed an \emph{LL system}.
Due to invariance, 
the dynamics reduce to the Lie algebra $\g$, or on its dual $\g^*$ if working with the momentum formulation.

The reduced Lagrangian $L:\g\rightarrow \R$ defines the inertia operator $\I:\g \rightarrow \g^*$ by the
relation
\begin{equation*}
L(\xi)= \frac{1}{2}\langle \I\xi , \xi \rangle, \qquad \mbox{for} \;\; \xi \in \g,
\end{equation*}
where $\langle \cdot , \cdot \rangle$ denotes the duality pairing. The reduced Hamiltonian, $H:\g^*\rightarrow \R$,
is then given by
\begin{equation*}
H(\mu)= \frac{1}{2}\langle \mu , \I^{-1} \mu \rangle, \qquad \mbox{for} \;\; \mu \in \g^*.
\end{equation*}

The constraints can be expressed  as the annihilator of independent fixed co-vectors $\nu_i\in \g^*$. We say that an instantaneous velocity $\xi\in \g$ satisfies the constraints if
\begin{equation}
\label{E:LL-constraints}
\langle \nu_i , \xi \rangle =0, \qquad i=1, \dots, n,
\end{equation}
where $n$ is the number of constraints. The constraints are nonholonomic if the set of vectors $\xi \in \g$ satisfying the above condition
do not span a subalgebra of $\g$.

The reduced EPS equations on $\g^*$  are given by, see e.g. \cite{Blochbook},
\begin{equation}
\label{E:Euler-Poincare-Suslov}
\dot \mu=\ad^*_{\I^{-1}\mu} \mu + \sum_{i=1}^n\lambda_i \nu_i,
\end{equation}
where the multipliers $\lambda_i$ are certain scalars that are uniquely determined by the condition that the constraints (\ref{E:LL-constraints})
are satisfied.

\subsection*{Underwater rigid body with a left-invariant nonholonomic constraint}

We will be interested in the case $G=SE(3)$ with the coalgebra $\se(3)^*=({\bf k}, {\bf p})$ and $n=1$, which corresponds to
the motion of an underwater rigid body subject to a linear, left invariant and nonholonomic constraint
\begin{equation} \label{c1}
{\bf a}\cdot  \vecom  +  {\bf F}\cdot  {\bf V} =0,
\end{equation}
${\bf a},{\bf F}$ being some constant vectors in the \emph{body} frame.
The constraint is nonholonomic provided that the set of vectors $(\vecom, {\bf V})$ that satisfy the above condition do not
form a subalgebra of $\se(3)$.

Then, in view of (\ref{E:Kirchhoff_1_3D}), the EPS equations (\ref{E:Euler-Poincare-Suslov}) become
\begin{equation}
\label{E:Kirchhoff_1_3D_constrained}
\begin{split}
\dot {\bf k}&= {\bf k} \times \vecom + {\bf p} \times {\bf V} + \lambda {\bf a}, \\
\dot {\bf p} &={\bf p} \times \vecom + \lambda {\bf F},
\end{split}
\end{equation}
where $({\bf k},{\bf p})=\I ( \omega,  {\bf V})$ as described by (\ref{E:Explicit_Legendre}),
the total inertia operator is $\I=\I_\mathcal{B} + \I_\mathcal{F}$,
and the multiplier $\lambda$ is uniquely determined by the constraint \eqref{c1}. The system thus describes dynamics on the linear subspace
${\mathfrak d}\subset se(3)$ defined by \eqref{c1} or on its image in  $se(3)^*$.

When both vectors ${\bf a},{\bf F}$ are nonzero, it is difficult to present a mechanical interpretation of the constraint.

So, we consider two special cases:
\medskip

1) ${\bf F}={\bf 0}$. That is, the constraint is only on the angular velocity: ${\bf a}\cdot  \vecom =0$.
Then the equations (\ref{E:Kirchhoff_1_3D_constrained}) represent a hydrodynamic generalization of
the classical Suslov problem, see e.g., \cite{Suslov}. The latter describes the motion of a rigid body about a fixed point in  presence of this constraint.

If the origin of the body is at its mass center ($\Jn_\mathcal{B}=0$) and in the added masses the translational and rotational components are decoupled
($\Jn_\mathcal{F}=0$), then, according to \eqref{E:Explicit_Legendre},
$$
{\bf k}= (\mathcal{I}_\mathcal{B}+\mathcal{I}_\mathcal{F}) \vecom, \quad
{\bf p}= m{\bf V}  + \mathcal{M}_\mathcal{F}{\bf V},
$$
and the system (\ref{E:Kirchhoff_1_3D_constrained}) takes the closed form
\begin{equation*}
\begin{split}
\dot {\bf k}&= {\bf k} \times  (\mathcal{I}_\mathcal{B}+\mathcal{I}_\mathcal{F})^{-1}{\bf k}
+ {\bf p} \times ( m I  + \mathcal{M}_\mathcal{F})^{-1} {\bf p} + \lambda {\bf a}, \\
\dot {\bf p} &={\bf p} \times (\mathcal{I}_\mathcal{B}+\mathcal{I}_\mathcal{F})^{-1}{\bf k} .
\end{split}
\end{equation*}
If we interpret the momentum ${\bf p}$ as a vector fixed in space, the above system
has the same form as the equations of the Suslov problem in the quadratic potential (Clebsch--Tisserand) field
$U=\frac 12 {\bf p}\cdot A {\bf p}$, $A=( m I  + \mathcal{M}_\mathcal{F})^{-1}$, which
was studied in detail in \cite{Koz1, Tatarinov} under the assumption  that $\bf a$ is an eigenvector
of $(\mathcal{I}_\mathcal{B}+\mathcal{I}_\mathcal{F})^{-1}$, see also \cite{Jo3}. The latter
condition is important to guarantee the existence of an invariant measure (see the discussion below).

Apparently, for the general tensor $\I_\mathcal{B}+\I_\mathcal{F}$ the system (\ref{E:Kirchhoff_1_3D_constrained})
with the constraint ${\bf a}\cdot \vecom =0$ has not been studied before.


2) ${\bf a}={\bf 0}$. That is, the constraint is only on the linear velocity of the body,
and the system can be interpreted as an underwater version of the 3-dimensional Chaplygin sleigh, a rigid body moving in ${\mathbb R}^3$ under the
condition ${\bf F}\cdot {\bf V}=0$. This mechanical setting is regarded as an approximate model of an underwater vehicle with a big fin,
as it was mentioned in the Introduction.

\paragraph{Existence of an invariant measure.} It is natural to  ask whether the equations
\eqref{E:Kirchhoff_1_3D_constrained}, \eqref{c1} possess a smooth invariant measure.
This problem has been considered for general EPS equations on Lie algebras of compact groups by Kozlov \cite{Koz3},
by Jovanovi\'c \cite{Jo1} in the non-compact case, and by Zenkov and Bloch \cite{ZB} for systems with nontrivial shape space.
Following \cite{Jo1}, the necessary and sufficient condition for equations \eqref{E:Euler-Poincare-Suslov}, \eqref{E:LL-constraints}
to have such a measure in the case $n=1$ is that the constraint covector $\nu=\nu_1 \in \g^*$ satisfies
\begin{equation*}
\frac{1}{\langle \nu , \I^{-1}\nu \rangle} \, \operatorname{ad}^*_{\I^{-1}\nu}\nu + T = c\nu, \qquad c\in \R,
\end{equation*}
where $T\in \g^*$ is defined by the relation $\langle T, \xi \rangle = \operatorname{trace}(\operatorname{ad}_\xi)$, $\xi\in \g$.

Since the group $SE(3)$ is unimodular, in our case $T=0$, and in case of the generic constraint \eqref{c1}  the above condition becomes
\begin{equation}
\label{E:inv-measure-cond}
({\bf a}\times {\bf U} + {\bf F} \times {\bf W} \, , \,  {\bf F} \times {\bf U})=c \, ({\bf a}, {\bf F}), \qquad c\in \R,
\end{equation}
where $( {\bf U},  {\bf W}):= \I^{-1}({\bf a}, {\bf F})$.

This condition is easily analyzed
in the special cases that we considered (either ${\bf F}={\bf 0}$ or ${\bf a} ={\bf 0}$) under
 the assumption that both $\Jn_\mathcal{B}=0$ and $\Jn_\mathcal{F}=0$.
 One can then show that \eqref{E:inv-measure-cond} is equivalent to
 asking that ${\bf a}$ is an eigenvector of $(\mathcal{I}_\mathcal{B}+\mathcal{I}_\mathcal{F})^{-1}$ in the case ${\bf F}={\bf 0}$, or to the condition that ${\bf F}$ is
 an eigenvector of $A=(mI + \mathcal{M}_\mathcal{F})^{-1}$ in the case ${\bf a}={\bf 0}$.

\section{The hydrodynamic planar Chaplygin sleigh}

We now consider in detail the two-dimensional version of (\ref{E:Kirchhoff_1_3D_constrained}), which corresponds to the coalgebra
$\se(2)^*=({k}, {\bf p})$ and the nonholonomic constraint
\begin{equation*}
a  \omega +  F_1v_1+F_2v_2=0.
\end{equation*}
Assume that $a=0$ and choose, without loss of generality, $F_1=0$. Then the constraint takes the form $v_2=0$ and
the reduced equations of motion are
\begin{equation}
\label{E:Sleigh_2D_constrained}
\begin{split}
\dot k &= v_2p_1-v_1p_2, \\
\dot p_1 &= \omega p_2, \qquad
\dot p_2 = - \omega p_1 +\lambda .
\end{split}
\end{equation}
Here the column vectors $(k, {\bf p})^T$ and $(\omega, {\bf V})^T$ are related by $(k, {\bf p})^T=\I(\omega, {\bf V})^T$, with the
$3\times 3$ tensor $\I=\I_\B+\I_\F$. The multiplier $\lambda$ is determined from the condition $v_2=0$.

In the absence of fluid ($\I_\F=0$) the equations \eqref{E:Sleigh_2D_constrained} become
the classical Chaplygin sleigh problem, which goes back to 1911, \cite{Chaplygin}, and describes the motion of a planar rigid body
with a knife edge (a blade) that slides on the plane.
The nonholonomic constraint forbids the motion in the direction perpendicular to the knife edge. T
he asymptotic motions of the sleigh on the plane are straight-line uniform motions, see \cite{Blochbook, Chaplygin, NeFu}. 

In the presence of the potential fluid ($\I_\F\ne 0)$ the added masses define a more general kinetic energy, and
the system \eqref{E:Sleigh_2D_constrained} describes a hydrodynamic generalization of the Chaplygin sleigh.
We shall see that this leads to new features in the asymptotic behavior of the body.

\paragraph{The total inertia tensor.}
Introduce the body reference frame $\{ {\bf E_1} \, {\bf E_2} \}$ centered at the contact point between the knife edge and the plane and choose ${\bf E_1}$
to be parallel to the blade and ${\bf E_2}$ orthogonal to it. This ensures the above constraint $v_2=0$.

While the expression for $\I_\B$ with respect to the body frame was given in (\ref{E:Inertia_Matrices_2D})
for an arbitrary body, the expression for the tensor of adjoint masses $\I_\F$ can be given explicitly only for rather simple geometries.
A simple yet interesting case is an elliptical planar body with the semi-axes $A>B>0$.
Assume first that the origin is at the
center of the ellipse, but the coordinate axes ${\bf E_1}, \, {\bf E_2}$ are not aligned with the axes of the ellipse,
forming an angle $\theta$ (measured counter-clockwise), as illustrated in figure 1 (a). 

\begin{figure}[ht]
\centering
\subfigure[]{\includegraphics[width=5cm]{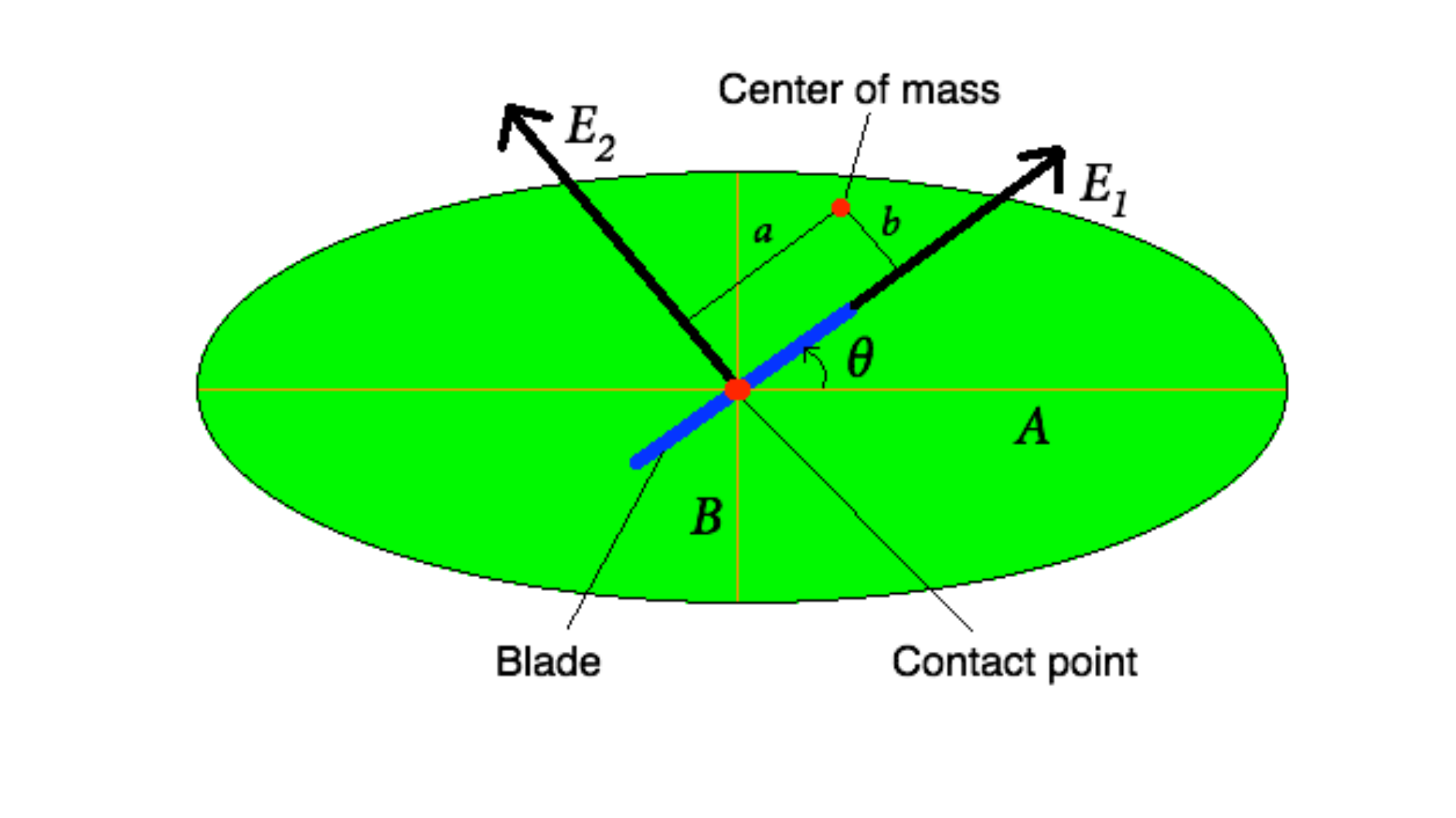}} \qquad
\subfigure[]{\includegraphics[width=4.85cm]{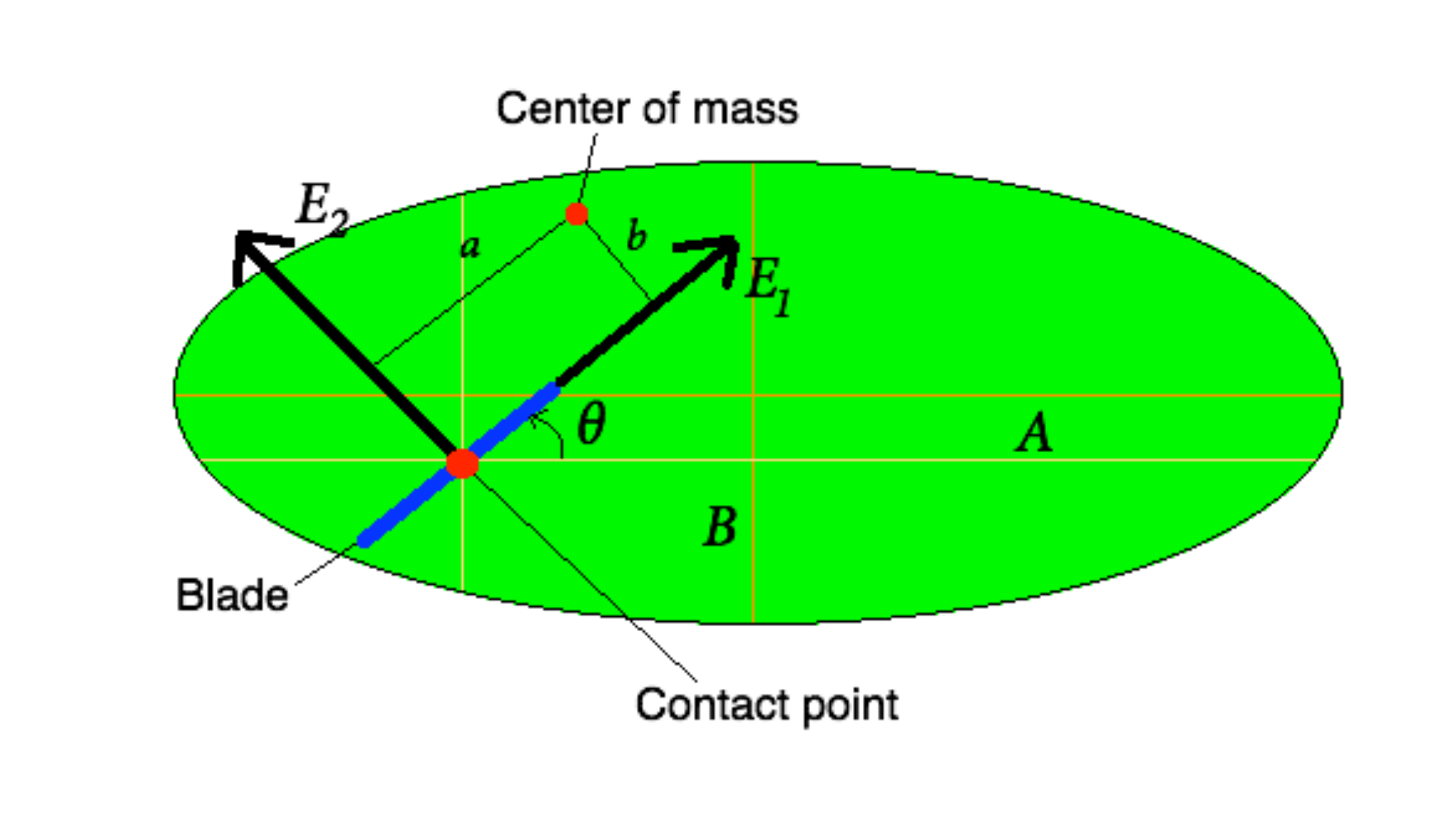}}
\caption{\small{The elliptical  sleigh. The blade makes an angle $\theta$ with the major axis of the ellipse.}}
\end{figure}

For this geometry, using the formula for the fluid potential given in \cite{Lamb}, one can show that
\begin{equation}
\label{E:Added_Inertia_Ellipse_Rot_Axes}
\I_\mathcal{F}=\rho \pi \left ( \begin{array}{ccc} \frac{ (A^2-B^2)^2}{4} & 0 & 0 \\ 0 &   B^2\cos^2\theta +A^2 \sin^2 \theta& \frac{ A^2-B^2}{2} \, \sin (2\theta) \\ 0 & \frac{ A^2-B^2}{2} \, \sin (2\theta) & A^2\cos ^2\theta + B^2\sin^2\theta \end{array} \right ).
\end{equation}
The total inertia tensor, $\I= \I_\mathcal{B} +\I_\F$, of the fluid-body system is then given by
\begin{small}
\begin{equation*}
\label{E:Total_Inertia_Ellipse_Rot_Axes}
\I= \left ( \begin{array}{ccc}  \In+m(a^2+b^2) + \rho \pi \frac{(A^2-B^2)^2}{4} & -mb & ma \\
-mb & m+\rho \pi \left (B^2\cos^2\theta +A^2 \sin^2 \theta \right )& \rho\pi \left (\frac{ A^2-B^2}{2} \right )  \sin (2\theta) \\
ma & \rho\pi \left (\frac{ A^2-B^2}{2} \right ) \sin (2\theta) & m+ \rho \pi (A^2\cos ^2\theta + B^2\sin^2\theta )\end{array} \right ).
\end{equation*}
\end{small}

Notice that in the presence of the fluid, if $\theta \neq n\frac{\pi}{2}, n\in \mathbb{Z}$, the coefficient $\I_{23}=\I_{32}$ is non-zero. This can never
be the case if the sleigh is moving in vacuum as one can see from the expression given for $\I_\B$ in  (\ref{E:Inertia_Matrices_2D}).
The appearance of this non-zero term leads to interesting dynamics that are studied below and that, to our knowledge, had not been described before in the literature.

Now, if the origin of $({\bf E_1} \, {\bf E_2})$ is not in the center of the ellipse (Figure 1 (b)),
then the tensor \eqref{E:Added_Inertia_Ellipse_Rot_Axes} takes a more general form with $(\I_{\cal F})_{13}, (\I_{\cal F})_{23}$ non-zero,
which can be calculated explicitly and lead to the corresponding modification of the total tensor $\I$.

In the sequel we assume that the shape of the sleigh is arbitrary convex and that its mass center does
not necessarily coincide with the origin, which leads to the general total inertia tensor
\begin{equation*}
\I= \left ( \begin{array}{ccc} J & -L_2 & L_1 \\ -L_2 & M &Z  \\ L_1& Z & N \end{array} \right ) .
\end{equation*}
We keep in mind that we expect to see new phenomena due to the presence of the fluid (when $Z\neq 0$). The tensor for classical Chaplygin sleigh is recovered by setting $Z=0$, $J=\In +m(a^2+b^2)$, $M=N=m$, $L_1=ma$, and $L_2=mb$.

Note that, in any case, since the total energy of the motion is positive definite, the tensor $\I$ has the same property.

\paragraph{Detailed equations of motion.}
The constraint written in terms of momenta is $v_2=\left( \I^{-1}(k, {\bf p} )^T \right)_3=0$.
Differentiating it and using (\ref{E:Sleigh_2D_constrained}), we find the multiplier
\begin{equation*}
\lambda=-\frac{1}{( { \I^{-1}} )_{33}} \left( \I^{-1} \left ( \begin{array}{c} v_2 p_1 - v_1 p_2 \\ \omega p_2 \\ -\omega p_1 \end{array} \right ) \right )_3,
\end{equation*}
where
\begin{equation*}
\I^{-1}= \frac{1}{\mbox{det}(\I)} \left ( \begin{array}{ccc} MN-Z^2 & ZL_1+NL_2 & -ZL_2-ML_1 \\ ZL_1+NL_2 & JN-L_1^2 &-L_1L_2-JZ  \\ -ZL_2-ML_1 & -L_1L_2-JZ  & JM-L_2^2 \end{array} \right ).
\end{equation*}

A long but straightforward calculation shows that,
by expressing $\omega, v_1$ and $v_2$ in terms of  $k, p_1, p_2$, substituting into (\ref{E:Sleigh_2D_constrained}),
and  enforcing the constraint $v_2=0$, one obtains:
\begin{equation}
\label{E:Working_Hydro_Sleigh_Equations}
\begin{split}
\dot \omega &=\frac{1}{D}\left (L_1 \omega + Z v_1 \right ) \left ( L_2 \omega - Mv_1\right ), \\
\dot v_1 &=\frac{1}{D} \left (L_1 \omega + Z v_1 \right ) \left (J\omega -L_2 v_1 \right ),
\end{split}
\end{equation}
where we set $D= \operatorname{det} (\I)( {\I^{-1}} )_{33}= MJ-L_2^2$. Note that $D>0$ since $\I$ and $\I^{-1}$ are positive definite.

The full motion of the sleigh on the plane is determined by the reconstruction equations \eqref{E:Reconstruction_Equations_2D}, which,
in our case with $v_2=0$, reduce to
\begin{equation} \label{recon}
\dot \phi = \omega, \qquad \dot x=v_1\cos \phi, \qquad \dot y = v_1\sin \phi.
\end{equation}

The reduced energy integral has
\begin{equation*}
H=\frac{1}{2}\left ( J\omega^2 +Mv_1^2 -2L_2\omega v_1 \right ),
\end{equation*}
and its level sets are ellipses on the $(\omega \, v_1)$-plane. As seen from the equations,
the straight line $\ell=\{L_1 \omega + Z v_1 =0\}$ consists of equilibrium points for the system.

Hence, if $L_1$ and $Z$ do not vanish simultaneously, the trajectories of (\ref{E:Working_Hydro_Sleigh_Equations}) are elliptic arcs that form heteroclinic connections
between the asymptotically unstable and stable equilibria on $\ell$ (see Fig. 2 (a)). 
The phase portrait is similar to that of the classical Chaplygin sleigh except that the line of equilibra is no longer  the $v_1$ axis if $Z\neq 0$.

\paragraph{Remark.} In fact, the reduced 2-dimensional system (\ref{E:Working_Hydro_Sleigh_Equations}) can be checked to be Hamiltonian
with respect to the following Poisson bracket of functions of $\omega, v_1$\footnote{
A similar observation for some other generalizations of the Chaplygin sleigh was made in \cite{Bor_Mam}.}
\begin{equation*}
\{F_1,F_2 \}:= -\frac{1}{D}\left (L_1 \omega + Z v_1 \right ) \left ( \frac{\partial F_1}{\partial \omega} \frac{\partial F_2}{\partial v_1} -
\frac{\partial F_1}{\partial v_1} \frac{\partial F_2}{\partial \omega} \right ).
\end{equation*}
The invariant symplectic leaves  consist of the semi-planes separated by the equilibria line $\ell$
and the zero-dimensional leaves formed by the points on this line. The above bracket can
be obtained using the construction developed in \cite{GN}.
\medskip

\paragraph{The case $Z=L_1=0$.} We start by considering  the most degenerate case when both $L_1$ and $Z$ vanish.
In particular, for the elliptical sleigh this is the case when $\theta=0 $ and $a=0$ in figure 1 (a).
Then all the solutions of \eqref{E:Working_Hydro_Sleigh_Equations} are equilibria.
Such a sleigh always performs a uniform circular or straight line motion whose parameters
depend on the initial condition. In this case, the euclidean measure (or any smooth measure for
that matter) is trivially preserved by equations \eqref{E:Working_Hydro_Sleigh_Equations}.
In fact this is the only case in which there is a smooth preserved measure as
we now show.

\begin{proposition}
The reduced equations \eqref{E:Working_Hydro_Sleigh_Equations} possess a smooth invariant measure if and only if $L_1=Z=0$.
\end{proposition}
\begin{proof}
The general condition \eqref{E:inv-measure-cond} for a preserved measure can be
specialized to our  two dimensional problem by putting $${\bf a}={\bf 0}, \quad {\bf F}=(0,1,0), \;\;
{\bf U}=\operatorname{det}(\I)^{-1}(-ZL_2-ML_1, 0, 0), \quad {\bf W}=\operatorname{det}(\I)^{-1}(-L_1L_2-JZ, JM-L_2^2, 0).$$
One gets the necessary and sufficient conditions
\begin{equation*}
ML_1+ZL_2=0, \qquad L_1L_2+JZ=0.
\end{equation*}
The above can be seen as a linear system of equations for $L_1$ and $Z$ with
non-zero determinant $MJ-L_2^2=D>0$. Hence this condition can only hold if $L_1=Z=0$.
\end{proof}

The result of this proposition is to be expected from the qualitative
behavior of the system that was described above in the case where $L_1$ and $Z$ do not
vanish simultaneously.

\paragraph{The case  $Z\ne 0$.} In this case the equilibrium points correspond to periodic circular motion of the body on the plane, and the contact point of the blade (the origin) goes
along circles of the radius
\begin{equation} \label{rad}
r= \left| \lim_{t\pm \infty} v_1 /\lim_{t\pm \infty} \omega \right| = \left|- \frac {L_1}{Z}\right | ,
\end{equation}
whereas the whole motion is an asymptotic evolution from one circular motion to the other one, but in opposite directions,
as shown in Fig. 2 (b). 
(Clearly, when $L_1=0$ the radius $r$ is zero, and the limit motions of the body are just rotations about the fixed origin.)

\begin{figure}[h]
\centering
\subfigure[Reduced phase portrait under the  assumption $L_1,Z>0$. The stable and unstable equilibra are represented by
filled and empty dots, respectively.]{\includegraphics[width=5cm]{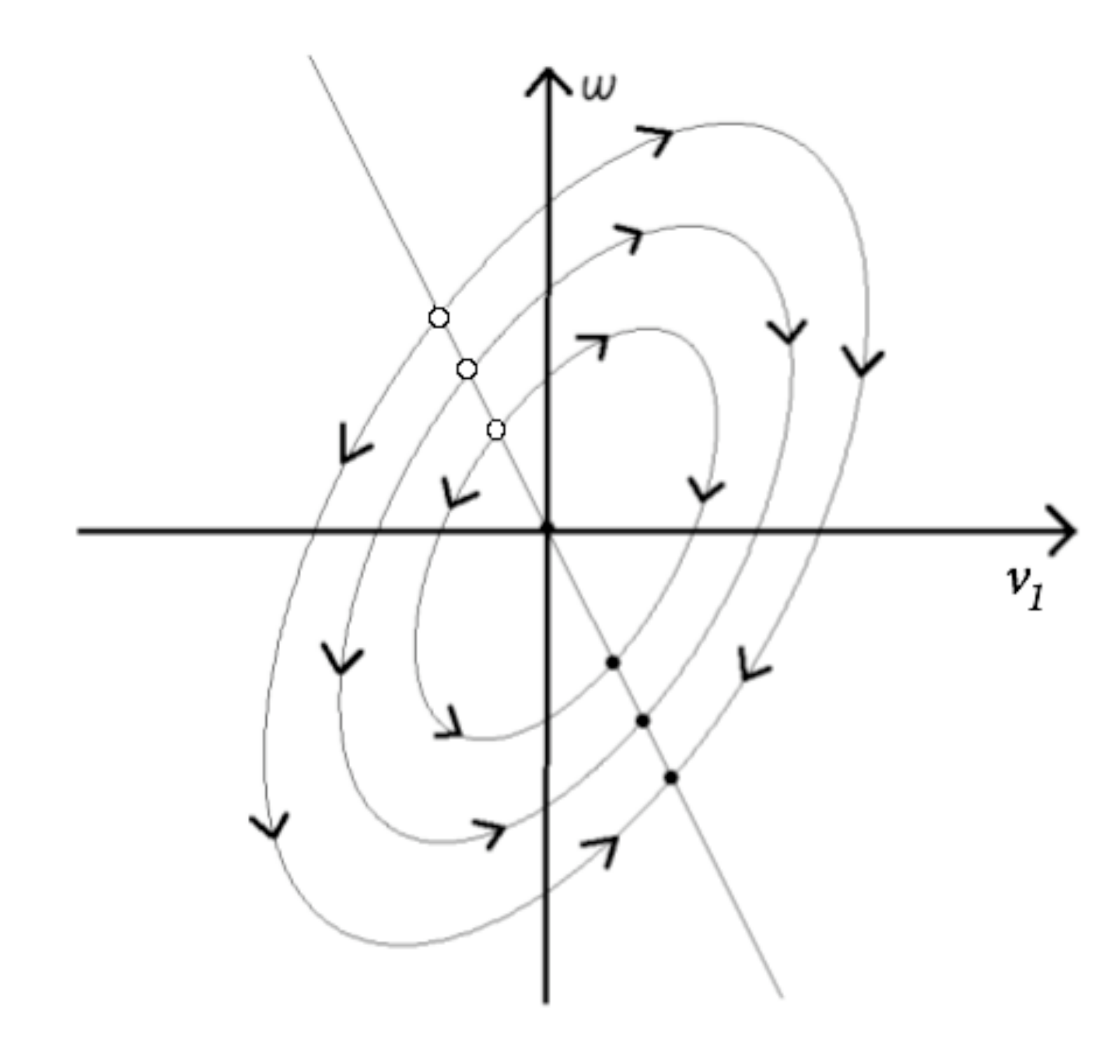}} \qquad \qquad
\subfigure[Trajectory of the elliptic sleigh on the plane. Asymptotic evolution from one circular motion to another one in opposite directions. The dot
on the sleigh surface represents its center of mass.]{\includegraphics[width=5cm]{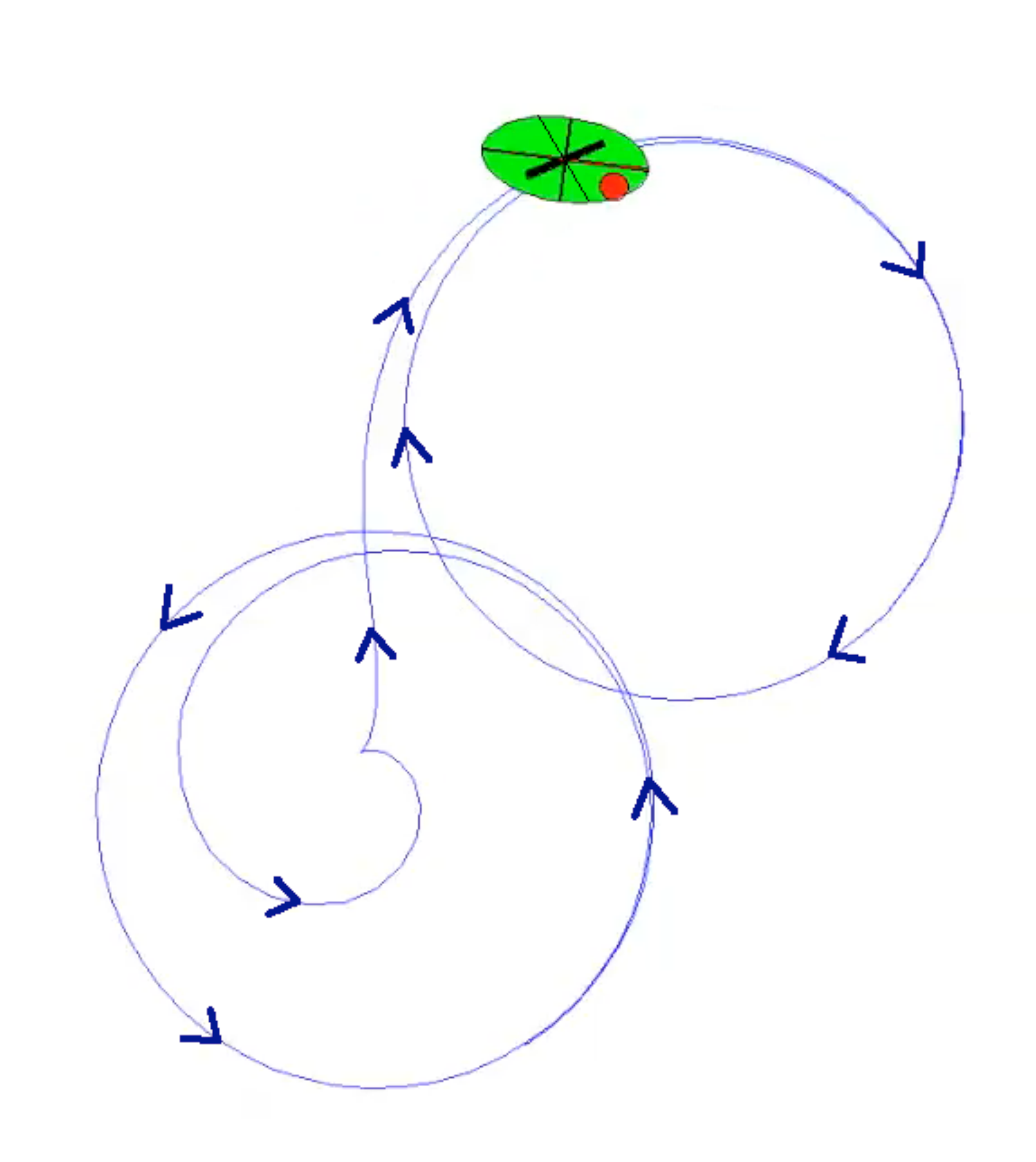}}
\caption{\small{Reduced phase portrait and trajectory of the sleigh in the plane.}}
\end{figure}

The preferred direction of rotation is determined by the following proposition
whose proof follows from a simple linear stability analysis.
%

\begin{proposition}
\label{P:Stability}
Let the line of equilibra $\ell= \{ L_1 \omega + Z v_1 =0\}$ be parameterized by $v_1(s)=L_1s, \ \omega(s)=-Zs, \; s\in \R$.
 The equilibra corresponding to $s<0$ are unstable, whereas the equilibra corresponding to $s>0$ are  stable.
\end{proposition}
\begin{proof}
We simply perform a linear stability analysis.
The matrix associated to the linearization about the equilibrium $v_1=-L_1s, \, \omega=Zs,$ is
\begin{equation*}
-s \, \left ( \begin{array}{cc} (L_2Z+ML_1)L_1    & (L_2Z+ML_1)Z   \\ (JZ+L_1L_2)L_1 &   (JZ+L_1L_2)Z \end{array} \right ).
\end{equation*}
It is seen that this matrix has eigenvalues $\lambda_1=0$ (corresponding to the continuum of equilibra along the
line  $\ell=\{L_1 \omega + Z v_1 =0\}$), and $\lambda_2=-sE$ with
\begin{equation}
\label{E:def_E}
E=JZ^2+2L_1L_2Z+ML_1^2= (Z,-L_1,0)\,\I\, (Z,-L_1,0)^T>0,
\end{equation}
since $\I$ is positive definite. Thus $\lambda_2$ is positive (negative) if $s<0$  ($s>0$), corresponding to the unstable (stable) direction.
\end{proof}

In particular, due to the above proposition and \eqref{rad}, the balanced elliptical sleigh depicted in fig 1 (a),
with $a>0$ and $0<\theta<\frac{\pi}{2}$, will have a limiting motion in the clockwise direction on a circle
of radius $r=\frac{2ma}{\rho \pi(A^2-B^2)\sin (2\theta)}$ as $t\rightarrow \infty$. Notice that
these conclusions on the qualitative asymptotic behavior of the system are independent of $b$, a feature that
is reminiscent of the classical unbalanced Chaplygin sleigh (see \cite{NeFu} and the discussion at the end of section 4). However, we shall see (Theorem 4.2) that
the distance between the centers of the circles does depend on $b$.
The conclusions of the proposition are also illustrated in figure 2.

%

\section{Explicit solution and asymptotic data}

The general solution of the reduced system \eqref{E:Working_Hydro_Sleigh_Equations} can be written in the form
\begin{equation} \label{sol}
\begin{split}
\omega(t) &= A ( \alpha \tanh (At) + \sigma c_{1} \sech (At)), \\
v_1 (t) & =A ( \beta \tanh (At) + \sigma c_{2} \sech (At)),\end{split}
\end{equation}
where the constants
\begin{gather} \label{pars}
\alpha = - \frac {D Z }{E}, \quad
\beta=\frac {D L_1}{E}, \quad
c_1 = \sqrt{D} \frac { ZL_2+ ML_1}{E}, \quad
c_2 =  \sqrt{D} \frac {L_1L_2 + ZJ }{E}, \\
E =JZ^2 +2 ZL_1L_2 +ML_1^2, \notag
\end{gather}
only depend on the components of the inertia tensor $\I$. Here $\sigma=\pm 1$ corresponding to the two different
branches of the trajectories on the phase portrait.

Notice that here the denominator $E>0$, as shown in \eqref{E:def_E} and that
the arbitrary constant $A\geq 0$ is related to the energy $H$ of the system by
\begin{equation*}
H=\frac{1}{2}\left ( \frac{D^2}{E} \right ) A^2.
\end{equation*}

\paragraph{The motion on the plane in the general case $Z\ne 0$.} In view of \eqref{recon}, the angle $\phi$ is calculated
by integrating the first expression in \eqref{sol}, which yields
\begin{gather}
\phi(t) = \phi_1 + \phi_2, \nonumber \\
\begin{aligned}
\phi_1 & = \int A \sigma c_1 \sech(At) \,  dt =2\sigma  c_1 (\arctan(e^{At}) -\pi/4), \\
\phi_2 & = \int A \alpha \tanh (At) dt= \alpha \ln (\cosh (At))+ \phi_0,
\end{aligned} \label{phi's}
\end{gather}
$\phi_0$ being an integration constant. The angles $\phi_1$ and $\phi_2$ are an odd and an even function of $t$
respectively. One can observe that they have quite different behavior:
\begin{equation} \label{lim_phi1}
\lim_{t\to \pm \infty}\phi_1= \pm \sigma\frac{c_1 \pi}{2}= \pm   \sigma \frac{\pi}{2} \frac{\sqrt{D} (Z L_2+M L_1)}{E},
\end{equation}
whereas,  as $t\rightarrow \pm \infty$, the angle $\phi_2$ asymptotically approaches the linear
function $l(t)= \pm A\alpha t +\phi_0 -\alpha \ln 2$.

The trajectory of the origin on the plane $(x,y)$ is then described by rather complicated integrals, which themselves do not say much about its properties.
However, it is natural to calculate the distance between the centers of the limit circles. To do this, we shall use

\begin{proposition} The centers of the limit circles coincide with the limit positions of the material point $C$, which in the body reference frame
$({\bf E_1} \, {\bf E_2})$ has fixed coordinates $(0,\beta/\alpha)=(0,-L_1/Z)$.
\end{proposition}

\begin{proof}
This can  be easily obtained from formula \eqref{rad}. Alternatively,
in view  of \eqref{sol}, \eqref{pars}, the $x$- and $y$-velocities of the point $C$ are
\begin{align*}
\dot x_C &= \left( v_1-\frac \beta \alpha \omega \right) \cos (\phi (t))= A \sigma \left ( \frac{\sqrt{D}}{Z} \right ) \sech (At) \, \cos(\phi (t)), \\
\dot y_C & = \left( v_1-\frac \beta \alpha \omega \right) \sin  (\phi (t))= A \sigma \left ( \frac{\sqrt{D}}{Z} \right ) \sech (At) \, \sin (\phi (t)),
\end{align*}
and both tend to zero as $t \to\pm \infty$. Hence $C$ coincides with the centers of the circumferences, since otherwise its limit velocity would not be zero.
\end{proof}

Setting $\phi=\phi_1+\phi_2$ as in \eqref{phi's}, we obtain the following expressions for the components of the vector
$(\Delta x_C, \Delta y_C)$ connecting the centers of the limit circumferences:
\begin{align*}
\Delta x_C & = \int_{-\infty}^\infty \dot x_C \, dt= \frac{\sigma\sqrt{D}}{Z}  \int_{-\infty}^\infty A\sech (At) ( \cos \phi_1 \cos\phi_2 - \sin \phi_1 \sin\phi_2 )\, dt, \notag \\
\Delta y_C & = \int_{-\infty}^\infty \dot y_C \, dt= \frac{\sigma\sqrt{D}}{Z}  \int_{-\infty}^\infty A\sech (At)  ( \sin \phi_1 \cos\phi_2 + \cos \phi_1 \sin\phi_2 )\, dt .
\end{align*}
Since $\phi_1$ is an odd and  $\phi_2$ is an even  function of time, the integrals are reduced to
\begin{align*}
\Delta x_C & =\frac{\sigma\sqrt{D}}{Z} \int_{-\infty}^\infty A\sech (At) \, \cos \phi_1 \cos\phi_2 \, dt,  \\
\Delta y_C & = \frac{\sigma\sqrt{D}}{Z}\int_{-\infty}^\infty A\sech (At)\, \cos \phi_1 \sin\phi_2 \, dt
\end{align*}
and, if we set $T=At$,
\begin{equation} \label{dels}
\begin{split}
\Delta x_C & =   \frac{\sigma \sqrt{D}}{Z} \int_{-\infty}^\infty  \sech (T) \,  \cos (2 c_1 (\arctan(e^T)-\pi/4 ) ) \,
\cos (\alpha \ln (\cosh T ) + \phi_0) \,\,  d T,  \\
\Delta y_C & =   \frac{\sigma \sqrt{D}}{Z} \int_{-\infty}^\infty \sech (T) \,  \cos ( 2c_1 (\arctan(e^T)-\pi/4 ) ) \,
\sin (\alpha \ln (\cosh T ) + \phi_0) \,\,  d T,\end{split}
\end{equation}
$c_1, \alpha$ being specified in \eqref{pars}.

It follows that, like the radii of the limit circumferences, {\it the distance $d$ between their centers does not depend on the energy, but only on the components
of the generalized inertia tensor}, as expected\footnote{
In view of the similarity of the reduced hydrodynamic Chaplygin sleigh and the nonholonomic Suslov problem, the distance $d$ can be regarded as an analog of the
angle between the axes of the limit permanent rotations of the Suslov rigid body in space.}.

The length scale of this distance is given by the ratio $\sqrt{D}/|Z|$ and
it depends parametrically on the dimensionless quantities $\alpha$ and $c_1$.
The  explicit dependence of the distance $d$ on the parameters $\alpha$ and $c_1$ is given by the following

\begin{theorem} \label{theorem_distance} The square of the distance $d$
between the limit circumferences is given by
\begin{equation}
\label{E:distance_sqd}
d^2=(\Delta x_C)^2 + (\Delta y_C)^2 = \frac{2\pi D}{Z^2} \left ( \frac{\alpha}{c_1^2+\alpha^2} \right ) \left ( \frac{\cosh ( \alpha \pi) - \cos (c_1 \pi )}{\sinh (\alpha \pi )} \right ).
\end{equation}

\end{theorem}

\begin{proof} Assume without loss of generality that $\phi_0=\alpha \ln 2$ in \eqref{phi's} and \eqref{dels}.
This restriction causes a rotation of the vector $(\Delta x_C, \Delta y_C)$, but does not
affect its length.

First, for simplicity, set $ c_1=0$, that is, $( {\I^{-1}} )_{13}=0$. 
Then, in view of \eqref{phi's}, $\phi_1\equiv 0$, and the integrals \eqref{dels} give
\begin{align}
\Delta x_C + i \Delta y_C & =  2 \frac{\sigma \sqrt{D}}{Z} \int_{-\infty}^\infty \frac {\exp(i \,\alpha \ln (e^{T}+e^{-T} )) } { e^{T}+e^{-T} } \, d T =\{z=e^T \}
\notag \\
 &= 2 \frac{\sigma\sqrt{D}}{Z} \int_{0}^\infty  \frac {\exp(i \, \alpha \ln (z+1/z  )) } { z^2+ 1} \, d z =\{u= \arctan z \}  \notag \\
&= 2 \frac{\sigma\sqrt{D}}{Z} \int_{0}^{\pi/2}  \exp  (i \, \alpha  [\ln (1+ \tan^2 u)- \ln (\tan u) ] ) \, d u \notag \\
& = 2 \frac{\sigma\sqrt{D}}{Z} \int_{0}^{\pi/2} \exp \left (i \, \alpha  \, \ln \left( \frac{\sec^2 u}{\tan u}\right)  \right ) \, d u \notag \\
&= 
2 \frac{\sigma\sqrt{D}}{Z} \int_{0}^{\pi/2}  (\cos u)^{-i \alpha} (\sin u)^{-i  \alpha} \, du \label{int_beta}
\end{align}
and, similarly,
\begin{equation} \label{int_beta1}
\Delta x_C - i \Delta y_C = 2 \frac{\sigma\sqrt{D}}{Z} \int_{0}^{\pi/2} (\cos u)^{i \alpha} (\sin u)^{i  \alpha} \, du .
\end{equation}

The last integrals have the form of the Euler {\it Beta}-function (see e.g., \cite{Erdelyi})
\begin{equation}
\label{E:def_beta}
B(x,y)= \int_0^1 t^{x-1} (1-t)^{y-1} dt= 2 \int_{0}^{\pi/2} (\cos u)^{2x-1} (\sin u)^{2y-1} du = \frac{\Gamma(x) \Gamma(y) }{\Gamma (x+y)},
\end{equation}
$\Gamma (\cdot)$ being the Euler {\it Gamma}-function with the properties
\begin{equation}\label{G_properties}
\Gamma \left(\frac 12 - z \right) \Gamma \left(\frac 12 + z \right)= \frac{\pi}{\cos(\pi z)}, \quad
\Gamma ( - z ) \Gamma (z) = - \frac{\pi}{z\sin (\pi z)}, \quad \Gamma (z+1)=z\Gamma (z), \quad z\in {\mathbb C} . 
\end{equation}
Then \eqref{int_beta}, \eqref{int_beta1} read
$$
\Delta x_C \pm i \Delta y_C = 2 \frac{\sigma\sqrt{D}}{Z} \frac 12 B ( 1/2 \mp i \alpha/2,  1/2 \mp i \alpha/2)
 = \frac{\sigma\sqrt{D}}{Z} \frac { \Gamma^2 (1/2 \mp i \alpha/2 )}{\Gamma (1 \mp i \alpha )}
$$
and, in view of \eqref{G_properties}, the square of the distance is
\begin{align}
d^2= & (\Delta x_C)^2+(\Delta y_C)^2  = \frac{D}{Z^2}
\frac{ \left[ \Gamma \left(\frac 12 - i \frac{\alpha}{2} \right) \Gamma \left(\frac 12 + i \frac{\alpha}{2} \right)\right]^2 }
 {\Gamma(1+i\alpha) \Gamma(1-i\alpha) } = \frac{D}{Z^2} \frac{\pi/\cos^2(\pi i \alpha/2) }{i \alpha/ \sin( \pi i \alpha) } \notag \\
 &= \frac{D}{Z^2} \frac{ 2 \pi \sin (\pi i \alpha )  }{i \alpha (1+ \cos (\pi i \alpha ) )}
= \frac{2 \pi D }{Z^2\alpha} \left ( \frac{\sinh (\pi \alpha)}{1+ \cosh (\pi \alpha)} \right )
= \frac{2 \pi D }{Z^2\alpha} \left ( \frac{ \cosh (\pi \alpha) -1}{\sinh (\pi \alpha)} \right ),  \notag
\end{align}
which is real and positive for real $\alpha$.

In the general case $c_1\ne 0$, under the same changes of variables, the integrals \eqref{dels} yield
\begin{align}
 \Delta x_C \mp i \Delta y_C &= 2 \frac{\sigma\sqrt{D}}{Z} \int_{0}^{\pi/2} \cos(2 c_1 (u-\pi/4))\, (\cos u)^{\pm i \alpha} (\sin u)^{\pm i \alpha} \, du  \notag \\
 &=2 \frac{\sigma\sqrt{D}}{Z}(2^{\mp i\alpha}) \int_{0}^{\pi/2} \cos(2 c_1 (u-\pi/4))\, (\sin (2u))^{\pm i \alpha} \, du =\, \{w=2u \} \notag \\
  &= \frac{\sigma\sqrt{D}}{Z}(2^{\mp i\alpha}) \int_{0}^{\pi} \cos( c_1 (w-\pi/2))\, (\sin w)^{\pm i \alpha} \, dw  \notag \\
   &= \frac{\sigma\sqrt{D}}{Z}(2^{\mp i\alpha}) \, \left [ \,  \cos \left ( \frac{c_1 \pi}{2} \right ) \int_{0}^{\pi} \cos( c_1 w )\, (\sin w)^{\pm i \alpha} \, dw \right . \,  \notag \\ & \qquad \qquad \qquad \qquad \qquad \qquad  \qquad  + \,   \left . \sin \left ( \frac{c_1 \pi}{2} \right ) \int_{0}^{\pi} \sin ( c_1 w )\, (\sin w)^{\pm i \alpha} \, dw \, \right ]
 .  \label{gen_ints}
\end{align}

The last two integrals in the right hand side can be calculated in terms of the Beta function by applying the general formulae \cite{GR}\footnote{
The original references for these formulae are \cite{LO} and \cite{WA}.}
\begin{equation*}
\begin{split}
\int_{0}^{\pi} \cos( c_1 w )\, (\sin w)^{\nu-1} \, dw &= \frac{\pi \cos \left ( \frac{c_1 \pi}{2} \right ) }{2^{\nu-1}\nu \, B \left( \frac{\nu +c_1 +1}{2}, \frac{\nu - c_1 +1}{2} \right )} \, , \\
\int_{0}^{\pi} \sin( c_1 w )\, (\sin w)^{\nu-1} \, dw &= \frac{\pi \sin \left ( \frac{c_1 \pi}{2} \right ) }{2^{\nu-1}\nu \, B \left( \frac{\nu +c_1 +1}{2}, \frac{\nu - c_1 +1}{2} \right )},
\end{split}
\end{equation*}
$\nu$ being a complex number with a positive real part. Then, after setting $\nu = 1\pm i\alpha$, \eqref{gen_ints} gives
\begin{align*}
 \Delta x_C \mp i \Delta y_C &= \frac{\sigma\sqrt{D}}{Z} \frac{4^{\mp i\alpha}\pi}{(1\pm i\alpha) \, B \left( 1+ \frac{c_1 \pm i\alpha }{2}, 1+ \frac{-c_1 \pm i\alpha }{2} \right )}.
 \end{align*}
Therefore, using \eqref{E:def_beta}, \eqref{G_properties}, we find that the square
of the distance is given by:
\begin{align*}
d^2 = (\Delta x_C)^2 + ( \Delta y_C)^2 &= \frac{D \pi ^2}{Z^2(1+\alpha^2)} \left [\frac{1}{ B \left( 1+ \frac{c_1+ i\alpha }{2}, 1+ \frac{-c_1+ i\alpha }{2} \right )\, B \left( 1+ \frac{c_1- i\alpha }{2}, 1+ \frac{-c_1- i\alpha }{2} \right )} \right ] \\
&=  \frac{D \pi ^2}{Z^2(1+\alpha^2)}  \left [\frac{\Gamma (2+i\alpha)\, \Gamma (2-i\alpha)}{ \Gamma \left( 1+ \frac{ c_1+i\alpha }{2} \right ) \, \Gamma \left ( 1+ \frac{-c_1+ i\alpha }{2}
\right )\, \Gamma \left( 1+ \frac{c_1- i\alpha }{2} \right ) \, \Gamma \left ( 1+ \frac{- c_1-i\alpha }{2} \right )} \right ] \\
&=  \frac{D \pi ^2}{Z^2(1+\alpha^2)}  \left [ \frac{\alpha^2(1+\alpha^2) \, \Gamma (i\alpha) \, \Gamma (-i\alpha)}{\left ( \frac{ c_1^2+\alpha^2 }{4}\right )^2 \, \left (  \Gamma \left(  \frac{ c_1+i\alpha }{2} \right ) \,  \Gamma \left( - \frac{c_1 + i\alpha}{2} \right )  \right ) \,  \left (  \Gamma \left(  \frac{ c_1-i\alpha }{2} \right ) \,  \Gamma \left( - \frac{c_1 - i\alpha}{2} \right )  \right )} \right ]  \\
&=  \frac{D \pi ^2\alpha^2}{Z^2}  \left [ \frac{ \,\sin   \left ( \pi \left (  \frac{ c_1+i\alpha }{2} \right ) \right )\, \sin   \left ( \pi \left (  \frac{ c_1- i\alpha }{2} \right ) \right )  }{ - i \alpha \pi \, \left ( \frac{ c_1^2+\alpha^2 }{4}\right ) \sin (  i \alpha \pi) } \right ] \\
&=  \frac{4\pi D }{Z^2} \left ( \frac{\alpha}{c_1^2+\alpha^2} \right )\left (
\frac{\sin^2 \left (\frac{c_1\pi}{2} \right ) \cosh^2\left (\frac{\alpha\pi}{2} \right ) + \cos^2 \left (\frac{c_1\pi}{2} \right ) \sinh^2\left (\frac{\alpha\pi}{2} \right )}{\sinh(\alpha \pi)} \right ),
 \end{align*}
and the last expression simplifies to \eqref{E:distance_sqd}.
\end{proof}

We mention that the formula of Theorem \ref{theorem_distance} is in perfect correspondence with numerical tests.

\paragraph{The case $Z=0, \; L_1\neq 0$.} The condition $Z=0$  corresponds to the absence of the fluid, or to the case when the blade is parallel to one of
two specific perpendicular directions in the body frame with the following property: if the solid is set in motion parallel to one of these, without rotation,
it will continue to move in this manner. For the elliptical sleigh, the two directions are precisely the principal axes of the ellipse.

In this case the system reduces to the classical Chaplygin sleigh, whose motion
on the plane was described in detail in \cite{NeFu}. Namely, this implies $\alpha=0$ and the line of equilibria on the phase plane $(\omega, v_1)$ is the axis
$\omega=0$. The trajectory of
the contact point on the plane in this case necessarily has a return point and, in view of \eqref{recon},  is given by the rather complicated integrals.
However, the limit behaviors of the sleigh are straight-line uniform motions and, according to \eqref{lim_phi1},
the angle between the limit lines is\footnote{Here $\Delta \phi$ is measured from the line of asymptotic straight line motion in the past
to the line of asymptotic straight line motion in the future in the trigonometric sense.}
$$
\Delta \phi = \sigma c_1 \pi = \sigma \pi \frac{\sqrt{D} (M L_1)}{E}.
$$
 Notice that $\Delta \phi$ and $\omega$ have the same sign (the latter does not change throughout the motion).

In the absence of the fluid, this expression becomes
$$\Delta \phi= \sigma \pi \frac{\sqrt{m ({\cal I}+ m a^2)}}{ma},$$
which does not depend on $b$. Figure 3 shows the trajectory of the sleigh in this case for different values of $\Delta \phi$ (see also \cite{NeFu}).
Note that $\left [ \frac{|\Delta \phi|}{2\pi} \right ]$ is the number of ``loops" that the sleigh performs in its transition between
the limit straight line motions.

\begin{figure}[h]
\centering
\subfigure[$0 < \Delta \phi < 2\pi$ ]{\includegraphics[width=4.5cm]{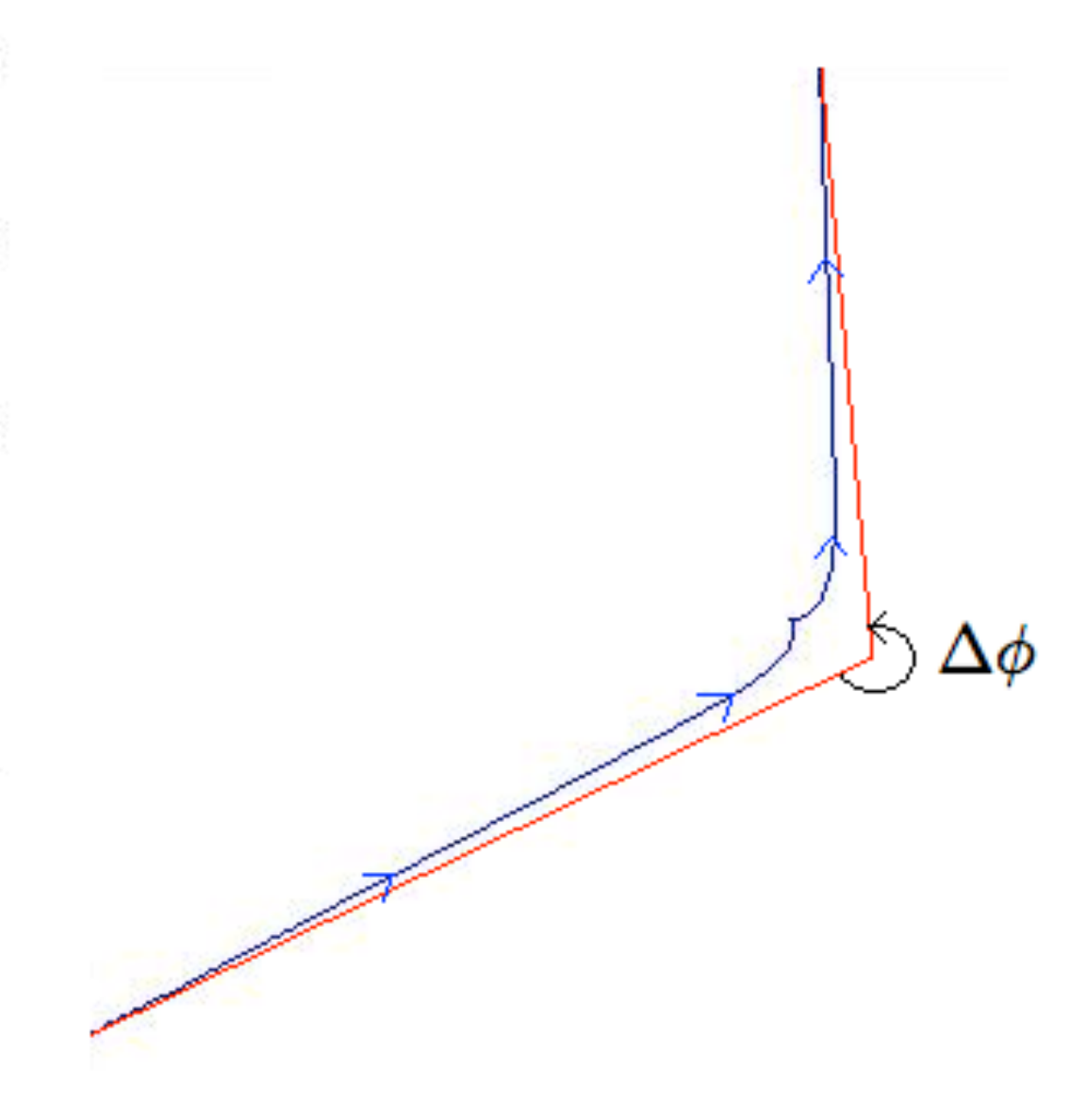}}  \; \;
\subfigure[$2\pi < \Delta \phi < 4\pi$ ]{\includegraphics[width=4.5cm]{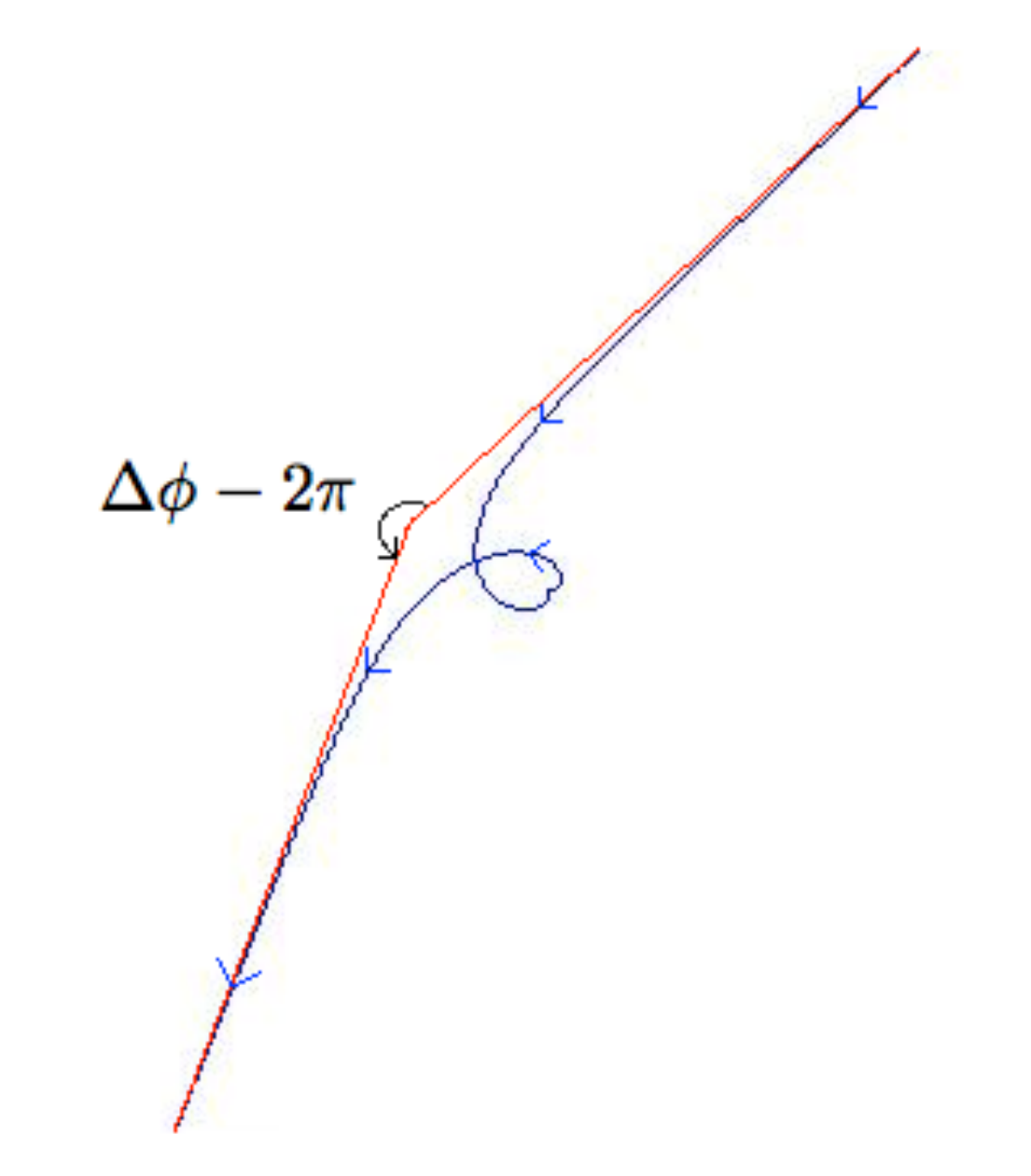}} \; \;
\subfigure[$4\pi < \Delta \phi < 6\pi$ ]{\includegraphics[width=6cm]{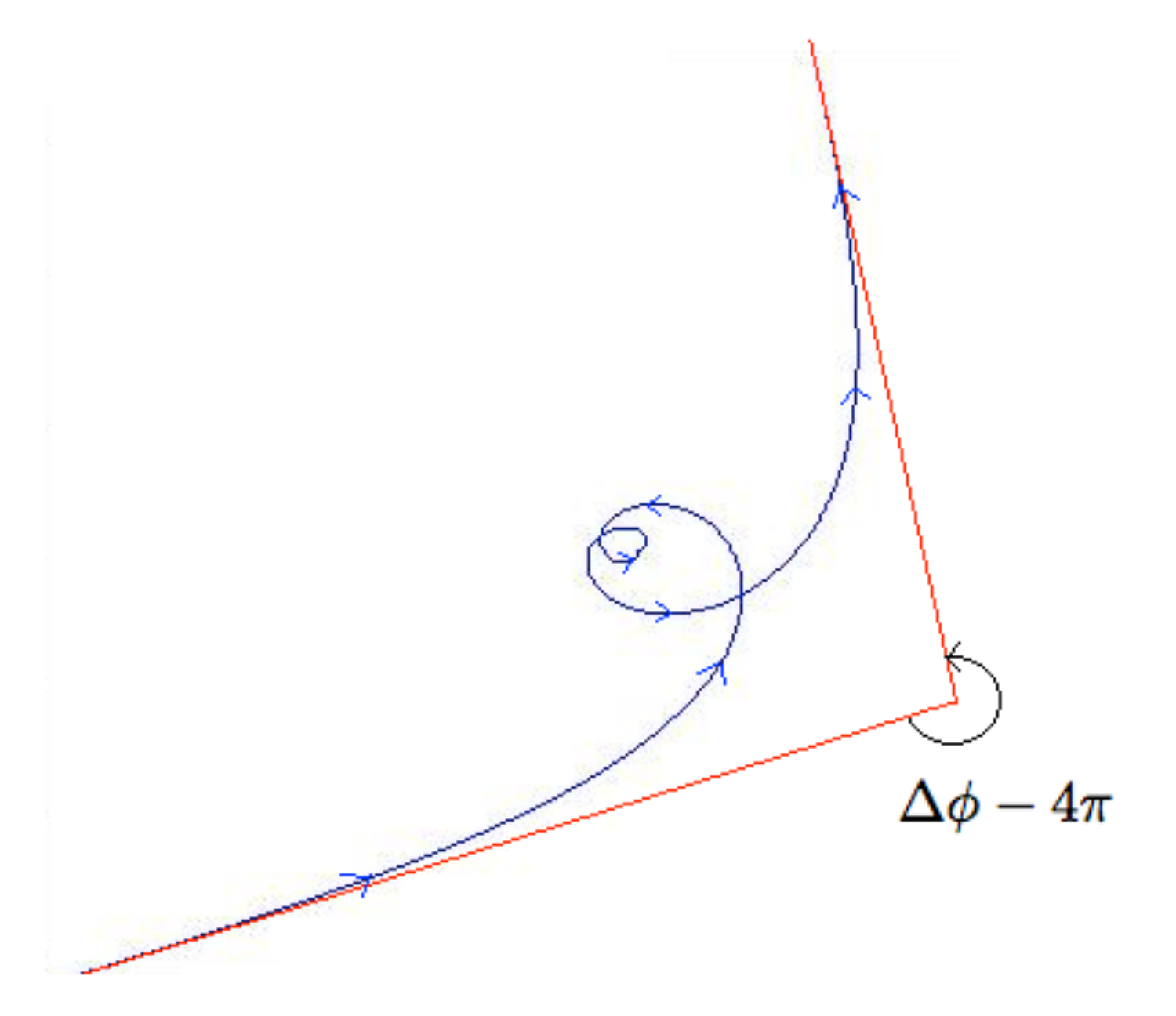}}
\caption{\small{Trajectory of the contact point of the sleigh on the plane in the case $Z=0$ for different values of $\Delta \phi$.
In all these cases $\omega>0$, so $\phi$ is an increasing function of
time and
the sleigh turns counterclockwise.}}
\end{figure}

%

\section*{Conclusions and further work}

A new series of examples of nonholonomic systems has been presented.
These are of interest from the point of view of applications in the design
of underwater vehicles and mechanisms, since, as we have mentioned, the constraint can be interpreted as a simple model for a fin.

From the mathematical point of view, our systems provide a motivation to study
the problem of nonholonomic geodesics on the group $SE(n)$ with a general left invariant kinetic energy metric.

The hydrodynamical version of the Chaplygin sleigh that has been considered, provides a new example
of a simple, integrable nonholonomic system with an interesting asymptotic behavior.
Its extensive analysis that we have presented can be of interest in the design of control mechanisms, see \cite{Osborne-Zenkov}.

It should be emphasized that we have only considered the case of zero circulation of the fluid.
Our preliminary studies show that in the presence of circulation the corresponding equations of motion are no longer of EPS type.
However, some of the features of the asymptotic motion remain.
In particular, for certain initial conditions, one has asymptotic evolution from one circular motion to another, as before,  but the radii of the limit circles are not the same.

 In this spirit, another interesting problem to consider is to couple the motion
 of the nonholonomic sleigh with point vortices. Such a  problem (without nonholonomic constraints) has received
 interest in the last years, see for example \cite{Ra, Bor_Mam_Ram, Va_et_al}.

 We hope to report with progress on the problems described above
 in the near future.

\paragraph{Acknowledgments} { \small Yu.F. acknowledges the
support of MCyT-FEDER grant MTM2006-00478 and grant MTM2009-12670 of the Spanish Ministry of Science and Technology and is grateful to the School of Mathematics of
Ecole Politechnique F\'ed\'erale de Lausanne for its hospitality during his recent stay there.
\newline
Luis G.N  acknowledges the hospitality at the Department de Matem\'atica Aplicada I, at UPC Barcelona for
his recent stay there.
\newline
We are also grateful to Irina Kukk for her help with checking the derivation of the formula of
Theorem \ref{theorem_distance} and to Maria Pzybulska,  Joris Vankerschaver, and Dmitry Zenkov for useful discussions.}

\end{document}